\pdfoutput=1
\documentclass[11pt]{article}
\usepackage{graphicx, verbatim}

\usepackage{amsfonts}
\usepackage{amssymb}
\usepackage{amstext}
\usepackage{amsmath}
\usepackage{amsthm,verbatim}
\usepackage{xspace}
\newcommand{\Prob}{\ensuremath{\mathbb{P}}}
\renewcommand{\Pr}{\ensuremath{\mathbb{P}}}
\newcommand{\E}{\ensuremath{\mathbb E}}
\newcommand{\R}{\ensuremath{\mathbb R}}

  \def\beq{\begin{eqnarray}} \def\eeq{\end{eqnarray}} \def\ben{\begin{enumerate}}
\def\een{\end{enumerate}} \def\bit{\begin{itemize}}
\def\bel{\begin{lemma}}
\def\eel{\end{lemma}}
\def\eit{\end{itemize}} \def\beqs{\begin{eqnarray*}} \def\eeqs{\end{eqnarray*}} \def\bel{\begin{lemma}} \def\eel{\end{lemma}}
\newcommand{\N}{\mathbb{N}}   \newcommand{\C}{\mathcal{C}}

 \newcommand{\I}{\mathbb{I}}   \newcommand{\p}{\mathbb{P}}
\newcommand{\PP}{\mathcal P}    
 \newcommand{\MM}{\mathcal M} \newcommand{\la}{\lambda}  
   \def\eps{{\epsilon}}  \def\ie{i.\,e.\,} 

\newcommand{\vol}[1]{\ensuremath{\mathrm{vol}\inp{#1}}}

\newcommand{\G}{\mathcal{G}}
\newcommand{\hit}{\mathcal{H}}

\usepackage[english]{babel}\usepackage[utf8]{inputenc}
\usepackage[noadjust]{cite}
\usepackage{etoolbox}
\usepackage{amsmath, amsthm}
\usepackage{amsfonts, amssymb}
\usepackage{bm}
\usepackage{ifthen}
\usepackage[boxed, linesnumbered]{algorithm2e}

\usepackage{letltxmacro}

\usepackage{nicefrac}
\usepackage{xstring}
\usepackage{pgf}
\usepackage{tikz}
\usetikzlibrary{arrows,automata}
\usepackage[T1]{fontenc}
\usepackage[tracking]{microtype}
\usepackage[ttscale=0.81]{libertine}
\usepackage[libertine]{newtxmath}

\renewcommand{\mathbf}[1]{\bm{#1}}

\usepackage{geometry}
\usepackage{fullpage}

\usepackage{pdfsync, graphicx}

\usepackage{framed}
\usepackage{marginnote}

\usepackage{todo}

\usepackage{booktabs,siunitx}

\usepackage{caption}
\usepackage{subcaption}
\usepackage{wrapfig}

\usepackage{hyperref}
\hypersetup{
  colorlinks=true,
  linkcolor=green!30!black,
  linktoc=all,
  citecolor=blue,
  bookmarksnumbered=true
}
\newcommand{\arxiv}[1]{\href{https://arxiv.org/abs/#1}{arXiv:#1}}
\usepackage{cleveref}

\newtheorem{theorem}{Theorem}[section]
\newtheorem{observation}[theorem]{Observation}

\newtheorem{fact}[theorem]{Fact}
\newtheorem{lemma}[theorem]{Lemma}

\newtheorem{proposition}[theorem]{Proposition}

\theoremstyle{definition}
\newtheorem{definition}[theorem]{Definition}

\crefname{theorem}{Theorem}{Theorems}
\crefname{observation}{Observation}{Observations}
\crefname{claim}{Claim}{Claims}
\crefname{condition}{Condition}{Conditions}
\crefname{example}{Example}{Examples}
\crefname{fact}{Fact}{Facts}
\crefname{lemma}{Lemma}{Lemmas}
\crefname{corollary}{Corollary}{Corollaries}
\crefname{definition}{Definition}{Definitions}
\crefname{remark}{Remark}{Remarks}
\crefname{proposition}{Proposition}{Propositions}

\newcommand{\st}[0]{\ensuremath{\;\mathbf{:}\;}}

\newcommand{\abs}[1]{\ensuremath{\left|#1\right|}}
\newcommand{\norm}[2][]{\ensuremath{\Vert #2 \Vert_{#1}}}

\newcommand{\ceil}[1]{\ensuremath{\left\lceil#1\right\rceil}}

\newcommand{\inb}[1]{\left\{#1\right\}}
\newcommand{\inp}[1]{\left(#1\right)}
\newcommand{\insq}[1]{\left[#1\right]}

\makeatletter
\newcommand*{\defeq}{\mathrel{\rlap{\raisebox{0.3ex}{$\m@th\cdot$}}\raisebox{-0.3ex}{$\m@th\cdot$}}=}
\newcommand*{\eqdef}{=
  \mathrel{\rlap{\raisebox{0.3ex}{$\m@th\cdot$}}\raisebox{-0.3ex}{$\m@th\cdot$}}}
\makeatother

\newcommand{\nfrac}[3][]{\nicefrac[#1]{#2}{#3}}

\newcommand{\poly}[1]{\ensuremath{\mathop{\mathrm{poly}}\inp{#1}}}

\theoremstyle{definition}

\theoremstyle{remark}

\numberwithin{equation}{section}

\begin{document}

\title{On the mixing time of coordinate Hit-and-Run}\author{Hariharan
  Narayanan\thanks{Email: \texttt{hariharan.narayanan@tifr.res.in}. TIFR Mumbai.}
  \and Piyush Srivastava\thanks{Email: \texttt{piyush.srivastava@tifr.res.in}. TIFR Mumbai.}}
\date{}
\maketitle
\begin{abstract}
  We obtain a polynomial upper bound on the mixing time
  $T_{CHR}(\epsilon)$ of the coordinate Hit-and-Run random walk on an
  $n-$dimensional convex body, where $T_{CHR}(\epsilon)$ is the number
  of steps needed in order to reach within $\epsilon$ of the uniform
  distribution with respect to the total variation distance, starting
  from a warm start (i.e., a distribution which has a density with
  respect to the uniform distribution on the convex body that is
  bounded above by a constant). Our upper bound is polynomial in
  $n, R$ and $\frac{1}{\epsilon}$, where we assume that the convex
  body contains the unit $\Vert\cdot\Vert_\infty$-unit ball $B_\infty$
  and is contained in its $R$-dilation $R\cdot B_\infty$.  Whether
  coordinate Hit-and-Run has a polynomial mixing time has been an open
  question.
\end{abstract}

{\let\thefootnote\relax
  \footnotetext{HN and PS acknowledge support from the Department of Atomic
    Energy, Government of India, under project no. RTI4001, the Ramanujan
    Fellowship of SERB, and the Infosys foundation, through its support for the
    Infosys-Chandrasekharan virtual center for Random Geometry.  PS acknowledges
    support from Adobe Systems Incorporated via a gift to TIFR.  The contents of
    this paper do not necessarily reflect the views of the funding agencies listed
    above.}
}

\thispagestyle{empty}

\newpage
\setcounter{page}{1}

\section{Introduction}

\subsection{Background}
Approximate uniform sampling from convex bodies is an important algorithmic
primitive and has consequently been the focus of a large body of research.  The
first algorithm for this problem that ran in time polynomial in the ambient
dimension of the convex body was given in the seminal work of Dyer, Frieze and
Kannan~\cite{DFK91}, and was based on a nearest-neighbor random walk on a
discrete grid inside the convex body.  This resulted in the first polynomial
time (in the ambient dimension and $1/\epsilon$) algorithm for obtaining a
$(1 \pm \epsilon)$-factor approximation for the volume of an input convex body
$\R^n$, given access to a so-called ``well-guaranteed membership
oracle''\footnote{A convex body $K$ is said to have a ``well-guaranteed
  membership oracle'' if it has a membership oracle, and further, if there are
  given positive numbers $r < R$ such that the body contains a ball (in an
  appropriate norm) of radius $r$ and is contained in a ball (in the same norm)
  of radius $R$.  The run-time of the algorithm of Dyer, Frieze and Kannan, and
  of many convex sampling algorithms that appeared later, is polynomial also in
  the ratio $R/r$.} for the body.

Following the success of the Markov chain Monte Carlo approach for the problem
in the work of Dyer, Frieze and Kannan, there have been many attempts at
designing other random walks for sampling from convex bodies with provably better
mixing times; for a partial list, see \cite{chen2018fast, cousins2017efficient,
  kannan1997random, kannan2012random, Laddha, lee2017geodesic,
  lee2017convergence, L99, LS93, HN2, Mangoubi}.  The goal in this line of work
is twofold.  First, the Markov chain should be such that each step of the random
walk can be implemented easily, under reasonable regularity conditions on the
convex body, and second, the number of steps the resulting Markov chain takes to
come $\epsilon$-close to the stationary distribution (in total variation
distance) should be bounded by a small polynomial in $n$ (the ambient
dimension), $1/\epsilon$, and the ratio of the in-radius and the diameter of the
body.

One such random walk is the so-called ``Hit-and-Run'' walk, first proposed in
the work of Smith~\cite{Smith}.  Given a convex body $K$, the samples
$x_1, x_2, \dots$ of this walk are generated as follows. Given $x_i = x$, choose
a direction $u$ on the unit sphere $\mathbb{S}^{n-1}$ uniformly at random, and
let $\ell$ be the unique chord of $K$ through $x$ along the direction $u$.
Then, pick a random point $y$ from the uniform measure on this chord and set
$x_{i+1} = y$.  A simplified version of this is the \emph{coordinate
  Hit-and-Run} (CHR) random walk, in which instead of sampling the direction $u$
from the unit sphere, it is sampled instead from the discrete uniform
distribution on the canonical basis vectors $e_1, e_2, \dots, e_n$.  This
restriction can make the CHR random walk easier to implement than the original
Hit-and-Run walk in certain settings~\cite{haraldsdottir2017chrr}.

Lovász~\cite{L99} analyzed the mixing time of the Hit-and-Run random walk from a
``warm start'', i.e., a starting distribution whose Radon-Nikodym derivative
with respect to the uniform distribution on $K$ is bounded above by some given
quantity $M$, and showed that with such a starting distribution, the Hit-and-Run
walk mixes to within $\epsilon$ total variation distance of the uniform
distribution on a convex body $K$ in time $\poly{n, R, M}$ (here, it is assumed
that $K$ is contained with the Euclidean ball of radius $R$ and contains within
it the standard Euclidean ball of radius $1$).  However, despite its
applicability in practice~\cite{haraldsdottir2017chrr,Fallahi20}, establishing a
similar polynomial mixing time bound on for the coordinate Hit-and-Run walk has
remained an open problem (see for example, \cite[pp.~23--24]{LeeVemp_KLS}).

\subsection{Our contribution}
In this paper, we prove that under regularity conditions similar to those in
\cite{L99}, the coordinate Hit-and-Run walk has a polynomial mixing time from a
warm start.  More formally, let $B_\infty$ denote the unit ball
$\inb{x \in \R^n \st \norm[\infty]{x} \leq 1}$ in $\norm[\infty]{\cdot}$. As
usual, for any positive real $r$, the notation $r\cdot B_\infty$ denotes the set
$\inb{r\cdot x \st x \in B_\infty}$. For $M \geq 1$, we say that a probability
distribution $\mu$ is ``$M$-warm'' with respect to another probability
distribution $\pi$ if $\mu$ is absolutely continuous with respect to $\pi$ and
if the Radon-Nikodym derivative of $\mu$ with respect to $\pi$ is pointwise
bounded above by $M$.  Our main result is then the following (see
\cref{thm:main} for a more quantitative statement).

\begin{theorem}\label{thm:main-intro}
  Let $K \subset \R^n$ be a closed convex body such that
  $B_\infty \subseteq K \subseteq R\cdot B_\infty$, for some $R \geq 1$.  Let
  $\mu_0$ be a probability distribution on $K$ that is $M$-warm with respect to
  the uniform distribution $\pi_K$ on $K$.  Starting with the initial ``warm
  start'' distribution $\mu_0$, let $\mu_k$ denote the probability distribution
  of the point obtained after $k \geq 1$ steps of the coordinate Hit-and-Run
  walk for $K$.  Then, for any $\epsilon > 0$ and
  $k \geq \tilde{O}\inp{\frac{n^{7}R^{4}M^{4}}{\epsilon^{4}}}$, we have
  $d_{TV}(\mu_k, \pi_K) \leq \eps.$ (Here $d_{TV}$ denotes the total variation
  distance between probability measures.)
\end{theorem}

It follows that starting from a $M$-warm start, the mixing time
$T_{CHR}(\epsilon)$ of the coordinate Hit-and-Run random walk on the body $K$ is
bounded by a polynomial in $n, M, R$ and $1/\epsilon$.  The precise form of the
polynomial is given in \Cref{thm:main}. Noting that the unit cube in $n$
dimensions is contained in a Euclidean ball of diameter $\sqrt{n}$, this result
also implies a similar result for convex sets sandwiched using Euclidean balls,
with a slightly worse dependence on $n$ (i.e. $n^9$).

Concurrent and independent work by Laddha and Vempala~\cite{LV20} has a main
result that is similar to ours: they also give a polynomial mixing time from a
warm start for the CHR chain. Their results hinge upon directly proving a novel
isoperimetric inequality with respect to axis aligned moves for convex sets that
are cubes, and then extending this to arbitrary convex sets by tiling them using
cubes.  The bounds provided by Laddha and Vempala are in terms of sandwiching
ratios $S$ involving Euclidean balls, and have the form $\tilde{O}(n^9 S^2)$
rather than $\tilde{O}(n^7R^4)$ of this paper.  Depending on the body, one or
the other bound can be better - it is possible to construct ``tube-like'' bodies
for which $\tilde{O}(n^7R^4)$ is larger, and cube-like bodies with a small
sandwiching ratio for which $\tilde{O}(n^9 S^2)$ is larger.

There exists an affine transformation that can make the sandwiching ratio using
Euclidean balls as small as $O(n)$ (such as rotating to the isotropic position),
and even $\tilde{O}(\sqrt{n})$ if one is willing to discard an insignificant
amount of mass (see the discussion on page 410 of Lov\'{a}sz and Vempala
\cite{LoV06}).  However, we believe that the situations where
coordinate hit-and-run would be most useful are those where there is an
underlying set of basis vector directions that has a statistically relevant
interpretation.  Coordinate hit-and-run also readily lends itself to a random
walk on the lattice points of a convex set. However, rounding the body using
non-diagonal linear transformations would not preserve the axis
directions. Similarly, non-unnimodular linear transformations do not preserve
the lattice. For these reasons, we do not include a detailed discussion of
pre-processing the body through such affine transformations.

\subsection{Technical overview and discussion}
\label{sec:technical-overview}
In this section, we fix a convex body $K \in \R^n$, and assume that $R \geq 1$
is such that $B_\infty \leq K \leq R\cdot B_\infty$.  One of the main
ingredients of Lovász's proof of the fast mixing of the Hit-and-Run
walk~\cite{L99} is the following ``overlap property'' of the Hit-and-Run walk.
For $v \in K$, let $\hit_v$ denote the probability distribution of the point
sampled by executing one step of the Hit-and-Run walk on $K$ started from $v$.
Lovász~\cite{L99} showed that when $u, v \in K$ are ``close'', the probability
distributions $H_u$ and $H_v$ have total variation distance bounded away from
$1$.  The ``closeness'' condition was defined in~\cite{L99} in terms of a
combination of the Euclidean distance and the ``cross-ratio distance'' (derived
from the Hilbert metric on $K)$.  The second ingredient is an appropriate
isoperimetric inequality, which in the case of~\cite{L99} was a new
isoperimetric inequality in terms of the Hilbert metric.  These two ingredients
were then combined in~\cite{L99} to obtain lower bounds on the $s$-conductance
(see~\cref{eq:s-cond}) of the Hit-and-Run random walk for all $0 < s < 1/2$,
from which the mixing time bound from a warm start follows via a result of
Lovász and Simonovits~\cite{LS93} (see \cref{thm:mixing}).

Unfortunately, the overlap property fails for the coordinate Hit-and-Run
chain. For, denoting by $\C_v$ the probability distribution of the point sampled
by executing one step of the coordinate Hit-and-Run walk on $K$ started from
$v$, it is easy to see that whenever $x, y \in K$ are such that $x - y$ has a
non-zero component along each of the coordinate directions, $\C_x$ and $\C_y$
have disjoint supports, so that $d_{TV}(\C_x, \C_y) = 1$.  In order to
circumvent this difficulty, we consider in \cref{sec:gaussian-chain} a
multi-step variant of the coordinate Hit-and-Run chain, which we call the
Gaussian chain. We then establish an analogue of the overlap property (see
\cref{def:overlap} for a precise definition of the property in our setting) for
this chain.  The method for passing from this version of the overlap property to
a lower bound on $s$-conductance, via an isoperimetric inequality, is similar to
the one used by Lovász~\cite{L99}, and is outlined in
\cref{sec:overlaps-conductance}.  An isoperimetric inequality due to Lovász and
Simonovits~\cite{LS93} (see \cref{thm:isoperimetry}) is then combined with the
results in \cref{sec:overlaps-conductance}, and with a comparison argument
between the Gaussian chain and the actual coordinate Hit-and-Run chain, to give
a lower bound on the $s$-conductance of the latter (see
\cref{sec:s-cond-coord}).  The mixing time result follows immediately from this
lower bound on the $s$-conductance using \Cref{thm:mixing} due to Lovász and
Simonovits~\cite{LS93}; the details are provided in \cref{sec:final}.

The isoperimetric inequality with respect to the cross-ratio distance used
in~\cite{L99} was an important ingredient in the optimized run-time bounds
obtained therein for the Hit-and-Run walk.  An extension of this inequality was
later derived by Lovász and Vempala~\cite{lovasz2006hit}, who used it to prove
fast mixing for the Hit-and-Run walk while substantially relaxing the
requirement of ``warm start'' used in~\cite{L99}. Unfortunately, the coordinate
Hit-and-Run walk does not appear to share with the Hit-and-Run walk the nice
geometric properties that enabled the use of these improved isoperimetric
inequalities.  Optimizing the mixing time bounds obtained for coordinate
Hit-and-Run in this paper, and relaxing the ``warm start'' condition, both
therefore remain interesting open problems.

\section{Preliminaries}
\label{sec:preliminaries}

\paragraph{Notation} We will denote positive universal constants that are less
than $1$ by $c$ and positive universal constants that are greater than $1$ by
$C$.

\subsection{The coordinate Hit-and-Run walk}
Let $K\subseteq \R^n$ be a convex
body. Let $e_1, \dots, e_n$ denote the canonical basis of $\R^n$.
The law of the coordinate Hit-and-Run (CHR) walk can then be expressed as
follows.  Fix $v$ belonging to the interior of $K$.  To implement one step of the
CHR walk, do the following.  First, choose $i \in [n]$ uniformly at random, and
let $\ell$ be the cord $K \cap (v + e_i \R).$ Then, choose $u$ uniformly at
random from $\ell$ (with respect to the $1-$dimensional Lebesgue measure
$\lambda_1$ on $\R$).  Thus, the transition probability densities (with respect
to $\la_1$) are given by
\begin{equation*}
  \PP_{vu} =
  \begin{cases}
    \frac{1}{n|K \cap (v + e_i \R)|} & \text{$u \in K \cap (v + e_i \R)$ for some $i \in [n]$,}\\
    0 & \text{ otherwise.}
  \end{cases}
\end{equation*}

\subsection{Markov schemes and isoperimetry}
In this section we collect some definitions and results about Markov chains that
will be used in our proofs.  We mostly follow the terminology and notation of
Lovász and Simonovits~\cite{LS93}.

For any measurable $S \subseteq \R^n$, we denote by $\vol{S}$ its Lebesgue
volume, and by $\pi_S$ the uniform probability distribution on $S$.  Thus,
$\pi_S(A) = \nfrac{\vol{S \cap A}}{\vol{S}}$ for any measurable $A$.  Fix a
measurable set $S \subseteq \R^n$ of finite volume and let $\mathcal{S}$ denote
the set of measurable subsets of $S$. A \emph{Markov scheme} on $S$ is a map
$P: S \times \mathcal{S} \rightarrow [0,1]$ such that (i) for every $u \in S$,
$P(u, \cdot)$ is a probability measure on $S$, and (ii) for every
$A \in \mathcal{S}$, the function $P(\cdot, A) : S \rightarrow [0,1]$ is
measurable.  Given a Markov scheme $P$ and a probability distribution $\mu_0$ on
$S$, A \emph{Markov chain} $M = M(\mu, P)$ is a sequence
$X_0, X_1, \dots, X_k,\dots$ of random points in $S$ such that $X_0 \sim \mu_0$
and for $i \geq 1$, $X_i \sim P(X_{i-1}, \cdot)$.  We denote by $\mu_{k, P}$ the
probability distribution of the random variable $X_k$.

A probability measure $Q$ on $S$ is said to be \emph{stationary} with respect to
a Markov scheme $P$ on $S$ if for all measurable $A\subseteq S$,
\begin{displaymath}
  \int_{S}P(u, A)\; Q(du) = Q(A).
\end{displaymath}
Such a $Q$ is said to be \emph{reversible} with respect to $P$ if for all
measurable $A, B \subseteq S$,
\begin{displaymath}
  \int_{A}P(u, B) \; Q(du) = \int_{B}P(u, A)\;Q(du).
\end{displaymath}
Note that if $Q$ is reversible with respect to $P$, then it is also stationary
with respect to $Q$.

The \emph{ergodic flow} $\Phi_{P,Q}$ of a Markov scheme $P$ on $S$ with respect
to a probability distribution $Q$ on $S$ is defined as
\begin{displaymath}
  \Phi_{P,Q}(A) \defeq \int_A P(u, S - A) Q(du),\text{ for all measurable $A \subseteq S$.}
\end{displaymath}
Note that when $Q$ is stationary with respect to $P$, one can easily check
that~\cite[p.~366]{LS93}
\begin{equation}
  \label{eq:symmetry-of-flow}
  \Phi_{P, Q}(A) = \Phi_{P, Q}(S - A).
\end{equation}
By abuse of notation, we extend the definition of ergodic flow to pairs of measurable sets:
\begin{displaymath}
  \Phi_{P,Q}(A, B) \defeq \int_A P(u, B) Q(du),\text{ for all measurable $A, B \subseteq S$.}
\end{displaymath}
The \emph{$s$-conductance}~\cite[p.~367]{LS93} $\Phi_s$, where $0 \leq s < 1/2$,
of a Markov scheme $P$ on $S$ with respect to a probability distribution $Q$ on
$S$ is defined as
\begin{equation}
  \Phi_s \defeq \inf_{A \st s < Q(A) \leq 1/2} \frac{\Phi_{P,Q}(A)}{Q(A) - s}.\label{eq:s-cond}
\end{equation}
Quantitative estimates involving the conductance of Markov chains have been an
important ingredient of various results proving mixing time upper bounds for Markov
chains.  In our proof, we will use the following estimate, stated in this
form by Lovász and Simonovits~\cite{LS93}.

\begin{theorem}[{\cite[Corollary 1.6]{LS93}}]\label{thm:mixing}
  Let $S$ be a measurable subset of $\R^n$ of finite volume. Let $P$ be a Markov
  scheme on $S$, and let $\mu_0$ and $Q$ be probability distributions on $S$
  such that $Q$ is stationary with respect to $P$. Assume that $Q$ is atom-free,
  i.e., $Q(\inb{u}) = 0$ for all $u \in S$. Fix $0 < s < 1/2$, and define
  \begin{displaymath}
    H_s \defeq \sup_{A \subseteq S \st  Q(A) \leq s} \abs{\mu_0(A) - Q(A)}.
  \end{displaymath}
  For each $k \geq 0$, let $\mu_{k, P}$ denote the distribution of the $k$th
  element in the Markov chain $(\mu_0, P)$.  Then, for all non-negative integers
  $k$,
  \begin{displaymath}
    d_{TV}(\mu_{k, P}, Q) \leq \inp{ 1 + \frac{(1-\Phi_s^2/2)^k}{s}}\cdot
    H_s,
  \end{displaymath}
  where $\Phi_s$ is the $s$-conductance of $P$ with respect to $Q$.
\end{theorem}

We will also need the following isoperimetric inequality.
\begin{theorem}[{\cite[Corollary 2.7]{LS93}}]
  \label{thm:isoperimetry}
  Let $\delta > 0$ be arbitrary. Fix any norm $\norm[\ell]{\cdot}$ on $\R^n$.
  Let $K \subseteq \R^n$ be a convex body and let $K_1$ and $K_2$ be
  disjoint measurable subsets of $K$ such that for all $u \in K_1$ and
  $v \in K_2$, one has $\norm[\ell]{u - v} \geq \delta$.  Suppose that the
  diameter of $K$ in the $\norm[\ell]{\cdot}$-norm is at most $D$ (i.e.,
  $\norm[\ell]{x - y} \leq D$ for all $x, y \in K$).  Then,
  \begin{displaymath}
    \vol{K - \inp{K_1 \cup K_2}} \geq \frac{2\delta}{D - \delta}
    \cdot \min\inb{\vol{K_1}, \vol{K_2}}.
  \end{displaymath}
\end{theorem}

\subsection{Volume of the robust interior}
\label{sec:volume-robust-inter}

Given a closed convex body $K$, we define, for $r > 0$,
\begin{displaymath}
  K_r \defeq \inb{x \in K \st \forall v, \norm[\infty]{v} \leq r \implies x + v
    \in K}.
\end{displaymath}
The following proposition follows directly from \cite[Proposition 2]{DFK91}.

\begin{proposition}
  \label{prop:distance-l2}
  Suppose that $\epsilon \in (0, 1/\sqrt{n})$ and $K \subset \R^n$ is a closed
  convex body such that $B = B_2 \subseteq K$.  Then
  $(1-\epsilon\sqrt{n})K \subseteq K_{\epsilon}$.  In particular,
  $\vol{K_\epsilon} \geq (1-\epsilon\sqrt{n})^n\vol{K}$.
\end{proposition}

In our setting, we will need the following proposition which improves upon
\cref{prop:distance-l2} in the case when $K$ is guaranteed to contain the
$\norm[\infty]{\cdot}$-ball $B_\infty$ instead of just the
$\norm[2]{\cdot}$-ball $B_2$.  The proof follows in an essentially identical
fashion to that of \cite[Proposition 2]{DFK91}, but we include it for
completeness.

\begin{proposition}
  \label{prop:distance-linfty}
  Suppose that $\epsilon \in (0, 1)$ and $K \subset \R^n$ is a closed convex
  body such that $B_\infty \subseteq K$.  Then
  $(1-\epsilon)K \subseteq K_{\epsilon}$.  In particular,
  $\vol{K_\epsilon} \geq (1-\epsilon)^n \vol{K}$.
\end{proposition}

\begin{proof}
  We only need to show that $\norm[\infty]{z - x} > \epsilon$ for all
  $x \in (1-\epsilon)K$ and $z \in K^c$.  Fix any $z \in K^c$.  Now, by the
  separating hyperplane theorem, there exists a non-zero vector $v$ such that
  $v^Tz > 1$ and $v^Ty \leq 1$ for all $y\in K$.  Since $B_\infty \subseteq K$,
  the latter inequality implies that $\norm[1]{v} \leq 1$.  Further, for any
  $x \in (1-\epsilon)K$, the same inequality gives $v^T{x} \leq (1-\epsilon)$.
  But then we have
  \begin{displaymath}
    \epsilon < v^T(z - x) \leq \norm[1]{v}\cdot\norm[\infty]{z - x} \leq
    \norm[\infty]{z - x}. \qedhere
  \end{displaymath}
\end{proof}

\section{Overlaps and conductance}
\label{sec:overlaps-conductance}

The next definition and the results following it in this section are based on a
method proposed by Lovász~\cite[Section 5]{L99}.

\begin{definition}[$(\epsilon, \delta, \nu)$-overlap property]\label{def:overlap} Let $K$ be a
  convex body in $\R^n$.  A Markov scheme $P$ on $K$ is said to have the
  $(\epsilon, \delta, \nu)$-overlap property with respect to $K'$, where $K'$ is
  a convex subset of $K$, if (i)
  $\vol{K'} \geq \vol{K}(1 - \epsilon)$ and (ii) for all $u, v \in K'$
  satisfying $\norm[2]{u - v} \leq \delta$, we have
  \begin{displaymath}
    d_{TV}(P(u, \cdot), P(v, \cdot)) \leq 1 - \nu.
  \end{displaymath}
\end{definition}

\begin{lemma}\label{lem:Kr}
  Let $\epsilon, \delta, \nu \in (0, 1/2)$.  Let $K$ be a convex body in $\R^n$
  of $\norm[2]{\cdot}$-diameter at most $D \geq 2\delta$ and let $P$ be a Markov
  scheme on $K$ which is reversible with respect to the uniform distribution
  $\pi_K$ on $K$.  Suppose also that $P$ has the
  $(\epsilon, \delta, \nu)$-overlap property with respect to $K'$.  Then, the
  $\epsilon$-conductance of $P$ with respect to $\pi_K$ is at least
  $\frac{\nu\delta}{4(D-\delta)}$. In fact, if $S_1, S_2$ is any arbitrary
  partition of $K$ into disjoint measurable subsets, and we denote $K'\cap S_1$
  by $T_1$ and $K'\cap S_2$ by $T_2$, then we have
  \begin{align*}
    \Phi_{P, \pi_K}(S_1) \geq \max \inb{\Phi_{P, \pi_K}(S_1, T_2),
    \Phi_{P, \pi_K}(S_2, T_1)
    }
    & \geq
      \left(\frac{\nu\delta}{4(D-\delta)}\right)\min\inb{\pi_K(T_1),
      \pi_K(T_2)},\\
    & \geq
      \left(\frac{\nu\delta}{4(D-\delta)}\right)\min\inb{\pi_K(S_1) - \epsilon,
      \pi_K(S_2) - \epsilon}.
  \end{align*}
\end{lemma}
\begin{proof}
  The proof follows the same method as that of Lovász~\cite[Sections 5 and
  6]{L99}.  Let $S_1, S_2$ be any arbitrary partition of $K$ into disjoint
  measurable subsets, and let the convex subset $K' \subseteq K$ be as in the
  definition of the $(\epsilon, \delta, \nu)$-overlap property, with
  $T_1 \defeq S_1 \cap K'$, $T_2 \defeq S_2 \cap K'$. Note that since
  $\pi_K(K') \geq 1 - \epsilon$, one has
  \begin{equation}
    \label{eq:2}
    \begin{aligned}
      \pi_K(T_1) = \pi_K(S_1 \cap K') &\geq \pi_K(S_1) - \epsilon\text{, and}\\
      \pi_K(T_2) = \pi_K(S_2 \cap K') &\geq \pi_K(S_2) - \epsilon.
    \end{aligned}
  \end{equation}
  Now, consider the following subsets of $T_1$ and $T_2$:
  \begin{equation}
    \label{eq:3}
    \begin{aligned}
      S_1' &\defeq \inb{x \in T_1 \st P(x, S_2) < \nu/2}\text{, and}\\
      S_2' &\defeq \inb{x \in T_2 \st P(x, S_1) < \nu/2}.
    \end{aligned}
  \end{equation}
  It then follows from the definition of the overlap property that
  \begin{equation}
    x \in S_1', y \in S_2' \implies \norm[2]{x - y} > \delta,\label{eq:4}
  \end{equation}
  for if not, then for such $x$ and $y$ satisfying
  $\norm[2]{x - y} \leq \delta$, we would have
  \begin{align*}
    1 - \nu \geq d_{TV}(P(x, \cdot), P(y, \cdot)) \geq P(x, S_1) - P(y, S_1) = 1
    - P(x, S_2) - P(y, S_1) > 1 - \nu,
  \end{align*}
  which is a contradiction.  We can therefore apply the isoperimetric inequality
  of \cref{thm:isoperimetry} to the subsets $S_1'$ and $S_2'$ of the convex body
  $K'$ to get
  \begin{equation}
    \pi_K(K' - (S_1' \cup S_2'))
    \geq \frac{2\delta}{D -
      \delta}\min\inb{\pi_K(S_1'), \pi_K(S_2')}.\label{eq:7}
  \end{equation}

  Note that since $P$ is reversible with respect to $\pi_K$, we have
  $\Phi_{P, \pi_K}(S_1, T_2) = \Phi_{P, \pi_K}(T_2, S_1)$.  We thus have
  \begin{equation}
    \Phi_{P, \pi_K}(S_1, T_2)  = \int\limits_{T_2}P(u, S_1) \pi_K(du) \geq \int\limits_{T_2 - S_2'}P(u, S_1) \pi_K(du) \geq \frac{\nu}{2}\pi_K(T_2-S_2').\label{eq:5}
  \end{equation}
  Similarly, we get
  \begin{equation}
    \Phi_{P, \pi_K}(S_2, T_1)  = \int\limits_{T_1}P(u, S_2) \pi_K(du) \geq \int\limits_{T_1 - S_1'}P(u, S_2) \pi_K(du) \geq \frac{\nu}{2}\pi_K(T_1-S_1').\label{eq:6}
  \end{equation}
  Consider now the condition:
  \begin{equation}
    \pi_K(S_1') > \frac{1}{2}\pi_K(T_1)\text{ and } \pi_K(S_2') >
    \frac{1}{2}\pi_K(T_2).
    \label{eq:9}
  \end{equation}
  The rest of the proof is divided into two cases:

  \noindent\textbf{Case 1: The condition in \cref{eq:9} is false.}  In this
  case, we must have $\pi_K(S_1') \leq \frac{1}{2} \pi_K(T_1)$ or
  $\pi_K(S_2') \leq \frac{1}{2}\pi_K(T_2)$.  Using \cref{eq:5} or
  \cref{eq:6} in the respective cases, we immediately get
  \begin{equation}
    \max\inb{\Phi_{P, \pi_K}(S_1, T_2),  \Phi_{P, \pi_K}(S_2, T_1)}
    \geq \frac{\nu}{4}\min\inb{\pi_K(T_1), \pi_K(T_2)}
    \stackrel{\text{\cref{eq:2}}}{\geq}
    \frac{\nu}{4}\min\inb{\pi_K(S_1) - \epsilon, \pi_K(S_2)
      -\epsilon},\label{eq:8}
  \end{equation}
  and the required inequality follows because of the assumption
  $D \geq 2\delta$.

  \noindent\textbf{Case 2: The condition in \cref{eq:9} is true.} By adding the
  inequalities in \cref{eq:5,eq:6}, and using the fact that the sets $T_1$ and
  $T_2$ form a disjoint partition of $K'$, we have
  \begin{equation}
    \max\inb{\Phi_{P, \pi_K}(S_1, T_2),  \Phi_{P, \pi_K}(S_2, T_1)}
    \geq \frac{\nu}{4}\pi_K((T_1 - S_1') \cup
    (T_2 - S_2')) = \frac{\nu}{4}\pi_K(K' - (S_1' \cup S_2')).\label{eq:10}
  \end{equation}
  Now, using \cref{eq:9} followed by \cref{eq:2} in the isoperimetric inequality
  in \cref{eq:7}, we get
  \begin{align*}
    \pi_K(K' - (S_1' \cup S_2'))
    & \geq \frac{\delta}{D -
      \delta}\min\inb{\pi_K(T_1), \pi_K(T_2)}\\
    & \geq \frac{\delta}{D -
      \delta}\min\inb{\pi_K(S_1) - \epsilon, \pi_K(S_2) - \epsilon}.
  \end{align*}
  Substituting this into \cref{eq:10}, we again obtain the required inequality.
\end{proof}

\section{The Gaussian walk}
\label{sec:gaussian-chain}
In this section, we analyse an auxiliary Markov chain with Gaussian steps.  In
the next section, this chain will be compared with the CHR chain to establish
lower bounds on the $s$-conductance of the latter.

Let $N(0, \sigma^2)$ denote the Gaussian distribution on $\R$ with mean $0$ and
variance $\sigma^2$, and let $\gamma$ denote its density with respect to the
Lebesgue measure.  The Gaussian walk on a convex body $K$ is then
defined as follows.  For some $t \in \N$, suppose that we are given
$x_t = v \in K$. To generate $x_{t+1}$, do the following.  Choose $i \in [n]$
uniformly at random, and $\kappa \sim N(0, \sigma^2)$. Let $u = v + \kappa e_i$.
If $u \in K$, accept the proposed transition and set $x_{t+1} = u$. Else, the
proposal is rejected and $x_{t+1} = v$.

Thus, the transition probability densities $\G_{vu}$ (with respect to the
one-dimensional Lebesgue measure) are given by
\begin{equation*}
  \G_{vu} =
  \begin{cases}
    \frac{\gamma(|u - v|)}{n} & \text{ if $u \neq v$ and $u \in K \cap (v + e_i \R)$ for some $i \in [n]$,}\\
    0 & \text{ otherwise.}
  \end{cases}
\end{equation*}
In addition, the rejection probability, i.e., the probability that
$x_{t+1} = v$, conditional on $x_t = v$, is given by
$$ \p[x_{t+1} = v|x_t = v] = 1- \p[v + \kappa e_i \in K|x_t = v].$$
Note that like the coordinate Hit-and-Run chain, the Gaussian chain is also
reversible with respect to the uniform probability measure $\pi_K$ on $K$.

Given a positive integer $\tau$, let $\G_v^{(\tau)}$ denote the probability
distribution of $x_\tau$ given that $x_0 = v,$ and where, for
$0 \leq t \leq \tau - 1$, $x_{t+1}$ is generated from $x_t$ according to the
Gaussian chain $\G_{x_t}$.

We now set up some notation. Given a multi-index
$\I=(i_1, \dots, i_n) \in \N^n$, we define $\G_{v,\I}$ to be the unique Gaussian
distribution centered at $v \in \R^n$, whose covariance matrix $\Sigma_\I$ with
respect to the canonical basis is the diagonal matrix with $i_j \sigma^2$ as the
$j^{th}$ diagonal entry. We define $\MM_{n,\tau}$ to be the set of all such
multi-indices $\I = (i_1, i_2, \dots, i_n) \in \N^n$ for which
$\sum_j i_j = \tau$. We say that such an $\I$ is of \emph{full rank} if the
covariance matrix $\Sigma_{\I}$ has full rank, \ie, if all of the $i_k$ are at
least $1$.  Further, for $\I \in \MM_{n,\tau}$, we define
$\la_\I := {\tau \choose \I}n^{-\tau},$ where ${\binom{\tau}{\I}}$ is the
multinomial coefficient $\frac{\tau!}{\prod_k i_k!}.$ Observe that
\begin{equation}
\label{eq:addsup}
\sum_{\I \in \MM_{n,\tau}} \la_\I= 1.
\end{equation}
We then have the following lemma.
\begin{lemma}
  \label{lem:1.1} Fix $\sigma > 0$, $n \geq 2$ and
  $\tau = \tau(n) \defeq \lceil 20 n \ln n\rceil$.  Let $K \in \R^n$ be a convex
  body and let $v \in K$ be a point such that
  $\inf_{z \in \partial K} \|v - z\|_{\infty} > 100 \sigma \ln n.$ Define the
  measures $\mathcal{G}_v^{\tau}$ and $\mathcal{G}_{v, \I}$ as above, in terms
  of $N(0, \sigma^2)$.  Then,
  \begin{displaymath}
    d_{TV}\inp{\G_v^{(\tau)}, \sum_{\substack{\I \in \MM_{n,\tau}\\\I\text{ of full rank}}} \la_\I
      \G_{v,\I}} \leq 2 n^{-5}.
  \end{displaymath}
\end{lemma}
\begin{proof}
  Let $\mathcal{H}$ denote the probability measure
  $\left(\sum_{\I \in \MM_{n,\tau}} \la_\I \G_{v,\I}\right)$, and $\mathcal{H}'$
  the measure
  $\left(\sum_{\substack{\I \in \MM_{n,\tau}\\\I\text{ of full rank}}} \la_\I
    \G_{v,\I}\right)$.  Consider the following natural coupling between
  $\G_{v}^{(\tau)}$ and $\mathcal{H}$, for $0 \leq t \leq \tau$.  Set
  $x_0 = x_0' = v$.  For $0 \leq t \leq \tau - 1$, the pair
  $(x_{t+1}, x_{t+1}')$, given the pair $(x_t, x_t')$ is generated as follows.
  Pick $i \in [n]$ u.a.r. and $\kappa \sim N(0, \sigma^2)$, and set
  $x_{t+1}' = x_{t}' + \kappa e_i$.  Further $x_{t+1} = x_t + \kappa e_i$ if
  that point is in $K$, otherwise $x_{t+1} = x_{t}$.  As before, in the latter
  case we say that the update at time $t + 1$ was \emph{rejected}.  Let $E$ be
  the event that none of the updates at times $1\leq t \leq \tau$ are rejected.

  By definition, we have $x_\tau \sim \G_v^{(\tau)}$ and
  $x_\tau' \sim \mathcal{H}$.  Further, when the event $E$ occurs, we also have
  $x_\tau = x_\tau'$.  It therefore follows from standard coupling arguments
  that
  \begin{equation}
    d_{TV}(\G_v^{(\tau)}, \mathcal{H}) \leq \Pr[\lnot E] = 1 - \Prob[E].\label{eq:12}
  \end{equation}
  We now lower bound $\Pr[E]$.  For $1 \leq i \leq n$, let $A_i$ denote the
  number of times an update along the coordinate direction $i$ is proposed in
  the above coupling.  Then, by a standard Chernoff bound (see
  \cref{fct:chernoff-bound}), and recalling that $\tau = \ceil{20 n \ln n}$, we
  have
  \begin{equation}
    \label{eq:1}
    \Prob\insq{\exists i, 1\leq i \leq n, \text{ s.t. } A_i \geq 50\ln n} \leq n
    \cdot \exp\inp{-\tau/(3n)}\leq n^{-5}.
  \end{equation}
  We now recall that $x_0 = x_0' = v$ satisfies
  $\inf_{z \in \partial K} \|v - z\|_{\infty} > 100 \sigma \ln n$. Thus,
  conditioned on the event that each $A_i$ is at most $50\ln n$, standard
  Gaussian tail bounds imply that with high probability, each of the $x_j'$,
  $0 \leq j \leq \tau$, is actually contained in $K$, so that the event $E$
  occurs.  To establish this claim formally, let $A_{i,t}$, where
  $1 \leq i \leq n$ and $1 \leq t \leq \tau$, denote the number of times, up to
  time $t$, that an update along the coordinate direction $i$ is performed in
  the above coupling (thus, $A_i = A_{i, \tau}$).  Note that conditioned on the
  $A_{i,t}$, the $j$th coordinate $x_{s,j}'$ of $x_s'$ has the distribution
  $N(v_j, A_{j,s}\sigma^2)$.  Further, if we have
  $\vert x_{s,j}' - v_j\vert \leq 100 \sigma \ln n$ for all $1 \leq s \leq \tau$
  and $1 \leq j \leq n$, then the event $E$ occurs.  Using a union bound, we
  thus have
  \begin{multline*}
    \Pr[\lnot E \vert A_{i,t}, \text{ where } 1\leq i \leq n, 1 \leq t \leq
    \tau]\\
    \begin{aligned}
      &\leq \sum_{j = 1}^n\sum_{s = 1}^\tau \Pr[\vert x_{s,j}' - v_j\vert > 100
      \sigma \ln n \vert A_{i,t}, \text{ where } 1\leq i \leq n, 1 \leq t \leq
      \tau]\\
      &= \sum_{j = 1}^n\sum_{s = 1}^\tau \Pr_{Z \sim N(v_j, A_{j,s}\sigma^2)}[\vert Z -
      v_j\vert > 100 \sigma \ln n].
    \end{aligned}
  \end{multline*}
  Now, using the fact that $A_{j,s} \leq A_j$ for all $1 \leq j \leq n$ and
  $1 \leq s \leq \tau$, and then applying \Cref{fct:gaussian-tail-bound} to
  bound each term in the sum above, we get
  \begin{equation}
    \Pr[\lnot E \vert \forall 1\leq i \leq n, A_i \leq 50\ln n]
    \leq n\cdot\tau \cdot \exp\inp{-\frac{100^2\,\sigma^2(\ln n)^2}{100 \, \sigma^2 \ln
        n}} \leq n^{-90}.\label{eq:11}
  \end{equation}
  Together, \cref{eq:12,eq:1,eq:11} imply that
  \begin{equation}
    d_{TV}(\G_v^{(\tau)}, \mathcal{H})
    \leq \Pr[\lnot E] \leq
    1.5\, n^{-5}.\label{eq:15}
  \end{equation}
  Finally, we note that
  \begin{equation}
    \label{eq:16}
    d_{TV}({\mathcal{H'}, \mathcal{H}})
    \;\;\leq \sum_{\substack{\I \in \MM_{n,\tau}\\\I\text{ not of full rank}}}
    \la_\I \;\; \leq \;\; n\cdot\inp{1- \frac{1}{n}}^\tau \leq n^{-19}.
  \end{equation}
  The claim now follows from \cref{eq:15,eq:16} and the triangle inequality for
  $d_{TV}$.
\end{proof}

\begin{lemma}\label{lem:1.2}
  Let $\sigma > 0, n \geq 2, \tau = \tau(n)$ and the convex body $K \in \R^n$ be
  as in the statement of \cref{lem:1.1}.  Let $v, u \in K$ be points such that
  \begin{itemize}
  \item $\|v - u\|_2 \leq \sigma,$
  \item $\inf_{z \in \partial K} \|v - z\|_{\infty} > 100 \sigma \ln n,$ and
  \item $\inf_{z \in \partial K} \|u - z\|_{\infty} > 100 \sigma \ln n.$
  \end{itemize}
  Define the measures $\mathcal{G}_u^{(\tau)}, \mathcal{G}_v^{(\tau)}$ and
  $\mathcal{G}_{u, \I},\mathcal{G}_{v, \I}$ as above, in terms of
  $N(0, \sigma^2)$.  Then,
  \begin{displaymath}
    d_{TV}(\G_v^{(\tau)}, \G_u^{(\tau)}) \leq
    \frac{1}{2} + 4 n^{-5} \leq 3/4.
  \end{displaymath}
\end{lemma}
\begin{proof} Let us fix an $\I \in \MM_{n,\tau}$ such that $\I$ has full
  rank. As $\I$ has full rank, it follows that
  $\sigma^2I \preccurlyeq \Sigma_\I$ and
  $\Sigma_\I^{-1} \preccurlyeq \sigma^{-2}I$.  We then have
  \begin{equation}
    \label{eq:18}
    d_{TV}\inp{\G_{v, \I}, \G_{u, \I}} \leq \sqrt{\frac{1}{2}D_{KL}(\G_{v,
        \I}\Vert\G_{u, \I})} = \frac{1}{2}\sqrt{(v-u)^T\Sigma_{\I}^{-1}(v-u)} \leq
    \frac{1}{2\sigma}\norm[2]{v-u} \leq \frac{1}{2}.
  \end{equation}
  Here, the first inequality is Pinsker's inequality~(see \Cref{fct:Pinsker}),
  the equality is a direct calculation of the Kullback-Leibler divergence
  between two Gaussian distributions with the same covariance matrix, and the
  second inequality is from the observation above that
  $\Sigma_\I^{-1} \preccurlyeq \sigma^{-2}I$.  On the other hand, from
  \cref{lem:1.1} (applied separately to $\G_{u, \I}$ and $\G_{v, \I}$) and the
  triangle inequality for $d_{TV}$, we have
  \begin{displaymath}
    d_{TV}\inp{\G_{v, \I}, \G_{u, \I}} \leq  4\, n^{-5}
    + \sum_{\substack{\I \in \MM_{n,\tau}\\\I\text{ of full rank}}}
    \la_\I d_{TV}\inp{\G_{u,\I}, \G_{v,\I}}.
  \end{displaymath}
  The claim now follows from \cref{eq:18} since the sum of the $\lambda_\I$ over
  all $\I\in \MM_{n,\tau}$ of full rank is at most $1$.
\end{proof}

It follows quite easily from \Cref{lem:1.2} that the multi-step Gaussian chain
has the overlap property discussed in \Cref{sec:overlaps-conductance}.  In the
next section, we discuss this in more detail, and compare the Gaussian chain
with the CHR chain in order to obtain a lower-bound on the $s$-conductance of
the latter.

\section{The \texorpdfstring{$s$}{s}-conductance of coordinate Hit-and-Run}
\label{sec:s-cond-coord}
Our goal in this section is to the give a lower bound on the $s$-conductance of
the coordinate Hit-and-Run chain, by comparing it with the Gaussian chain
studied in the previous section.

Let the Markov scheme corresponding to coordinate Hit-and-Run be denoted $\C$.
We also carry forward the notation $\G$ and $\G^{(\tau)}$ for the Gaussian
scheme and its iterates, respectively, introduced in the previous section.

\begin{lemma} \label{lem:4.3} There are universal constants $C> 1$ and
  $0 < c < 1$ such that the following is true.  Let
  $\sigma > 0, n \geq 2, \tau = \tau(n) = \ceil{20n\ln n}$ be as in the
  statement of \cref{lem:1.1}, and let $K \in \R^n$ be a convex body such that
  for some $R \geq 1$, it is the case that
  $B_\infty \subseteq K \subseteq R\cdot B_\infty$. Set
  $\eps \defeq Cn\sigma\ln n$, and assume that $\sigma$ is chosen so that
  $\epsilon, \sigma < 1/2$. Define the chain $\G^{(\tau)}$ as in
  \cref{sec:gaussian-chain}, in terms of $N(0, \sigma^2)$.  Then, for any
  measurable set $S$ contained in $K$, such that $\eps < \pi_K(S)$ we have
  \begin{displaymath}
    \Phi_{\G^{(\tau)}, \pi_K}(S, K - S)
    \geq \left(\frac{c\sigma}{R\sqrt{n}}\right)\left(\min\inb{\pi_K(S), \pi_K(K - S)}
      - \eps \right).
\end{displaymath}
\end{lemma}
\begin{proof}
  Recall from \cref{def:overlap} that a Markov scheme $\PP$ on $K$ is
  said to have the $(\epsilon, \delta, \nu)$-overlap property with respect to
  $K'$ if $K'$ is a measurable convex subset of $K$ such that (i)
  $\vol{K'} \geq \vol{K}(1 - \epsilon)$, and further, (ii) for all $u, v \in K'$
  satisfying $\norm[2]{u - v} \leq \delta$, we have
  \begin{displaymath}
    d_{TV}(\PP(u, \cdot), \PP(v, \cdot)) \leq 1 - \nu.
  \end{displaymath}
  Set $C = 100$.  Now, from \cref{prop:distance-linfty}, we see that
  $\vol{K_{C\sigma\ln n}} \geq (1-C\sigma\ln n)^n\vol{K} \geq
  (1-\epsilon)\vol{K}$.  On the other hand, \cref{lem:1.2} implies that
  for any $u, v \in K_{C\sigma\ln n}$ such that $\norm[2]{u - v} \leq \sigma$,
  we have $d_{TV}\inp{\G^{(\tau)}(u, \cdot), \G^{(\tau)}(v, \cdot)} \leq 3/4$.
  It follows that that $\G^{(\tau)}$ has the $(\eps , \sigma, 1/4)-$overlap
  property with respect to $K_{C\sigma \ln n}$.

  Now, we apply \cref{lem:Kr} to the Markov scheme $\G^{(\tau)}$ to see that
  (note that the $\norm[2]{\cdot}$-diameter of $K$ is at most $2R\sqrt{n}$)
  \begin{displaymath}
    \Phi_{\G^{(\tau)}, \pi_K}(S, K - S) \geq
    \left(\frac{\sigma}{32 R\sqrt{n}}\right)\left(\min\inb{\pi_K(S), \pi_K(K - S)} -
      \eps \right).
  \end{displaymath}
  Setting $c = 1/32$ gives us the lemma.
\end{proof}

In order to use the above result for obtaining a lower bound on the
$s$-conductance of the coordinate Hit-and-Run chain, we need the following
simple observation comparing the ergodic flows of the \emph{single step}
coordinate Hit-and-Run chain and the \emph{single step} Gaussian chain (note
that what was analyzed in the previous lemma was an iterate of the latter).

\begin{observation} \label{lem:4.4} Let
  $\sigma > 0, n \geq 2, \tau = \tau(n) = \ceil{20 n \ln n}, R \geq 1$ and the
  convex body $K \in \R^n$ be as in the statement of \cref{lem:4.3}.  Define the
  Gaussian scheme $\G$ as in \cref{sec:gaussian-chain}, in terms of
  $N(0, \sigma^2)$.  Then, for any measurable subset $S$ of $K$, we have
  \begin{displaymath}
    \Phi_{\C, \pi_K}(S, K-S) \geq \frac{\sigma\sqrt{\pi}}{R\sqrt{2}} \cdot \Phi_{\G,
      \pi_K}(S, K - S) .
  \end{displaymath}
\end{observation}
\begin{proof}
  The measures $\G_x$ and $\C_x$ corresponding to transitions from any point
  $x \in S$ for the two schemes $\G$ and $\C$ respectively are absolutely
  continuous with respect to each other at all points $y \neq x$ (in particular,
  for all $y \in K - S$) in the support $\mathrm{supp}(\G_x)$ of $\G_x$.
  Further, the Radon-Nikodym derivative of $\C_x$ with respect to $\G_x$ at any
  $y \neq x$ in $\mathrm{supp}(\G_x)$ is bounded below by
  $(\sqrt{2\pi} \sigma/(2R))$. The claim now follows by integration.
\end{proof}

We now have all the ingredients to give a lower bound on the $s$-conductance of
the coordinate Hit-and-Run chain.

\begin{lemma}\label{lem:5.3}
  Let $n \geq 2, R \geq 1$ and the convex body $K \in \R^n$ be as in the
  statement of \cref{lem:4.3}.  Then, for any $s$ satisfying $0 < s < 1/2$, the
  $s$-conductance $\Phi_s$ of the coordinate Hit-and-Run chain $\C$ on $K$
  satisfies
  \begin{displaymath}
    \Phi_s \defeq \inf_{A \st s <
      \pi_K(A) \leq 1/2} \frac{\Phi_{\C,\pi_K}(A)}{\pi_K(A) - s} \geq
    \Omega\left(\frac{s^2 }{R^2n^{3.5} \ln^{3} n}\right).
  \end{displaymath}
  Here, the constants implicit in the $\Omega(\cdot)$ notation are independent
  of $n, R$ and $K$.
\end{lemma}
\begin{proof}
  Let the constants $C$ and $c$ be as in the statement of \Cref{lem:4.3}.  Set
  $\sigma = s/(Cn \ln n) \leq 1/2$, so that the quantity $\epsilon$ in
  \Cref{lem:4.3} satisfies $\epsilon = s < 1/2$.  Let
  $\tau(n) = \ceil{20n \ln n}$, and define the schemes $\G$ and $\G^{(\tau)}$ in
  terms of $N(0, \sigma^2)$, again as in the statement of \Cref{lem:4.3}.

   Let $S$ be an arbitrary measurable subset of $K$.  Consider the following
  experiment. Let $x_0$ be sampled from $\pi_K$, and let
  $x_0, x_1, \dots, x_\tau$ be a walk of length $\tau$ according to the scheme
  $\G$.  For $0 \leq i < \tau$, let $E_i$ denote the event that $x_i \in S$ and
  $x_{i+1} \in K - S$.  Similarly, let $E$ denote the event that $x_0 \in S$ and
  $x_\tau \in K - S$.  Note that $\Pr[E] \leq \sum_{i=0}^{\tau-1}\Pr[E_i]$.  On
  the other hand, since $\G$ is reversible with respect to $\pi_K$, we have
  $x_i \sim \pi_K$ for each $0 \leq i < \tau$, so that
  $\Pr[E_i] = \Phi_{\G, \pi_K}(S, K - S)$.  Similarly,
  $\Pr[E] = \Phi_{\G^{(\tau)}, \pi_K}(S, K - S)$.  Thus, we get that
  \begin{equation}
    \Phi_{\G, \pi_K}(S, K - S) \geq \frac{1}{\tau}\cdot \Phi_{\G^{(\tau)},
      \pi_K}(S, K - S).
  \end{equation}
  Combining this with \Cref{lem:4.4}, we therefore obtain
  \begin{equation}
    \Phi_{\C, \pi_K}(S, K - S) \geq \frac{c'\sigma}{\tau R} \Phi_{\G^{(\tau)},
      \pi_K}(S, K - S),\label{eq:4.8}
  \end{equation}
  where $c'$ is an absolute positive constant.

  Finally, applying Lemma~\ref{lem:4.3}, we than see that for any measurable
  subset $S \subset K$ such that $\pi(S) > s$, we have
  \begin{equation}
    \Phi_{\G^{(\tau)}, \pi_K}(S, K - S)
    \geq
    \left(\frac{c\sigma}{R\sqrt{n}}\right)
    \left(\min\inb{\pi_K(S), \pi_K(K-S)}  - s\right).
    \label{eq:17}
  \end{equation}
  Combining this with \cref{eq:4.8}, we see that for every measurable $S
  \subseteq K$  such that $s < \pi_K(S) \leq 1/2$, we have
  \begin{align}
    \Phi_{\C, \pi_K}(S, K - S)
    & \geq
      \left(\frac{c\cdot c' \cdot \sigma^2}{\tau \cdot R^2\sqrt{n}}\right)
      \left(\min\inb{\pi_K(S), \pi_K(K-S)}  - s\right) \nonumber\\
    & \geq \left(\frac{c'' s^2}{R^2 n^{3.5} \ln^3 n}\right)
      \inp{\pi_K(S)  -s},
      \label{eq:19}
  \end{align}
  where $c''$ is an absolute positive constant.  The claim follows immediately
  from \cref{eq:19}.
\end{proof}

\section{Mixing time of coordinate Hit-and-Run}
\label{sec:final}
We now restate and prove our main theorem (\Cref{thm:main-intro}).  Given the
lower-bound on the $s$-conductance of the coordinate Hit-and-Run scheme derived
in \cref{lem:5.3}, the proof follows immediately from \cref{thm:mixing}.
\begin{theorem}\label{thm:main}
  There exists a positive constant $C$ such that the following is true.  Fix
  $\epsilon \in (0, 1/2), n \geq 2$ and $R \geq 1$.  Let $K \subset \R^n$ be a
  closed convex body such that $B_\infty \subseteq K \subseteq R\cdot B_\infty$.
  Suppose that $\mu_0$ is absolutely continuous with respect to $\pi_K$ and that
  its Radon-Nikodym derivative with respect to $\pi_K$ is bounded above by
  $M \geq 1$.  Then, for any
  $k \geq \ceil{\frac{C M^4 R^4 n^7 \ln^6 n\ln(2M/\epsilon)}{\eps^4}}$, we
  have $d_{TV}(\mu_{k, \C}, \pi_K) \leq \eps.$
\end{theorem}
\begin{proof}
  Let $0 < s < 1/2$ be arbitrary.  We apply \cref{thm:mixing} to the coordinate
  Hit-and-Run scheme $\C$ and its stationary distribution $\pi_K$.  Let $H_s$ be
  as defined in \cref{thm:mixing}.  The absolute continuity constraint imposed
  on $\mu_0$ with respect to $\pi_K$ then implies that $H_s \leq M\cdot s$.  Let
  $\mu_k$ denote the distribution after $k$ iterations of $\C$ starting with the
  initial distribution $\mu_0$. \Cref{thm:mixing}, along with the bound on
  $\Phi_s(\C)$ obtained in \Cref{lem:5.3} then implies that (for some absolute
  constant $c$)
  \begin{displaymath}
    d_{TV}(\mu_{k, \C}, \pi_K) \leq Ms
    + M\exp\inp{\frac{-c
        ks^4 }{R^4n^7 \ln^6 n}}.
  \end{displaymath}
  Choosing $s = \epsilon/(2M)$, we then see that for any
  $k \geq \ceil{\frac{16 M^4 R^4n^7 \ln^6 n\ln(2M/\epsilon)}{c\eps^4}}$, we
  have $d_{TV}(\mu_{k, \C}, \pi_K) \leq \eps.$
\end{proof}

\providecommand{\bysame}{\leavevmode\hbox to3em{\hrulefill}\thinspace}
\providecommand{\MR}{\relax\ifhmode\unskip\space\fi MR }
\providecommand{\MRhref}[2]{\href{http://www.ams.org/mathscinet-getitem?mr=#1}{#2}
}
\providecommand{\href}[2]{#2}

\appendix

\section{Some standard inequalities}

In this section, we collect a few standard inequalities that we use at various
places in the main body of the paper.
\begin{fact}[\textbf{Chernoff bound, see, e.g.,~\cite[Theorem 4.1]{motwani-raghavan}}]
  Let $X_1, X_2, \dots, X_n$ be independent Bernoulli random variables, and let
  $X \defeq \sum_{i=1}^n X_i$. Define $\mu \defeq \E[X]$.  Then, for any
  $\delta > 0$,
\begin{displaymath}
  \Pr[X > (1+\delta)\mu] < \exp\inp{-\mu\cdot\inp{(1+\delta)\ln(1+\delta) - \delta}}.
\end{displaymath}
In particular, when $\delta = 1$, $\Pr[X > 2\mu] < \exp(-\mu/3).$\label{fct:chernoff-bound}
\end{fact}

\begin{fact}[\textbf{Gaussian tail bound}]\label{fct:gaussian-tail-bound}
  Fix $t \geq \sigma > 0$ and let $X \sim N(0, \sigma^2)$.  Then,
  \begin{displaymath}
    \Pr[\abs{X} \geq t] \leq \exp\inp{-\frac{t^2}{2\sigma^2}}.
  \end{displaymath}
\end{fact}
\begin{proof}
  Define $s \defeq t/\sigma \geq 1$. Let $Z \sim N(0, 1)$. Then,
  \begin{displaymath}
    \Pr[\abs{X} \geq t] = \Pr[\abs{Z} \geq s] =
    \sqrt{\frac{2}{\pi}}\int\limits_{s}^{\infty}\exp(-x^2/2)dx \leq
    \sqrt{\frac{2}{\pi s^2}}\int\limits_{s}^{\infty}x\exp(-x^2/2)dx
    = \sqrt{\frac{2}{\pi s^2}}\exp(-s^2/2).
  \end{displaymath}
  The claimed inequality follows since $s \geq 1$.
\end{proof}

\begin{fact}[\textbf{Pinsker's inequality, see, e.g.,~\cite[Theorem 4.19]{BLM}}]
  \label{fct:Pinsker}
  Let $P$ and $Q$ be probability measures on the same $\sigma$-field, such that
  $P$ is absolutely continuous with respect to $Q$.  Let $D_{KL}(P\Vert Q)$ denote
  the Kullback-Leibler divergence between $P$ and $Q$.  Then
  \begin{displaymath}
    d_{TV}(P, Q) \leq \sqrt{\frac{1}{2}D_{KL}(P\Vert Q)}.
  \end{displaymath}
\end{fact}
\end{document}